\documentclass[letterpaper, 11pt]{amsart}

\usepackage{amsthm}
\usepackage{amsmath}
\usepackage{amsfonts}
\usepackage{amssymb}
\usepackage{bbm}
\usepackage{hyperref} 
\usepackage{cleveref}
\usepackage{tikz-cd}

\newtheorem{proposition}{Proposition}
\newtheorem{theorem}{Theorem}
\newtheorem{corollary}{Corollary}
\newtheorem{lemma}{Lemma}
\theoremstyle{definition}

\newtheorem{remark}{Remark}

\newcommand{\range}[1]{\{1,\dots,#1\}}
\newcommand{\id}{\mathbbm{1}}
\newcommand{\card}[1]{|{#1}|}
\newcommand{\Q}{\mathbb{Q}}
\renewcommand{\L}{\mathcal{L}}
\DeclareMathOperator{\F}{F}
\DeclareMathOperator{\Hom}{Hom}
\DeclareMathOperator{\Aut}{Aut}
\DeclareMathOperator{\Ker}{Ker}

\DeclareMathOperator{\M}{M}
\DeclareMathOperator{\wt}{wt}
\DeclareMathOperator{\swc}{swc}
\DeclareMathOperator{\soc}{soc}
\DeclareMathOperator{\rad}{rad}

\title{When the extension property does not hold}
\author{Serhii Dyshko}
\email{dyshko@univ-tln.fr}

\begin{document}

\begin{abstract}
A complete extension theorem for linear codes over a module alphabet and the symmetrized weight composition is proved. It is shown that an extension property with respect to arbitrary weight function does not hold for module alphabets with a noncyclic socle.
\end{abstract}

\maketitle

\section{Introduction}
The famous MacWilliams Extension Theorem states that each linear Hamming isometry of a linear code extends to a monomial map. The result was originally proved in the MacWilliams' Ph.D. thesis, see \cite{macwilliams-phd61}, and was later generalized for linear codes over module alphabets.

Starting from the work \cite{ward-wood} of Ward and Wood, where they used the character theory to get an easy proof of the classical MacWilliams Extension Theorem, several generalizations of the result appeared in works of Dinh, L\'{o}pez-Permouth, Greferath, Wood, and others, see \cite{dinh-lopez-1, dinh-lopez, greferath, wood-foundations}. For finite rings and the Hamming weight, it was proved that the extension theorem holds for linear codes over a module alphabet if and only if the alphabet is pseudo-injective and has a cyclic socle, see \cite{wood-foundations}. 

Regarding the symmetrized weight composition, the extension theorem for the case of classical linear codes was proved by Goldberg in \cite{goldberg}.
There is a recent result in \cite{elgarem-megahed-wood} where the authors proved that if an alphabet has a cyclic socle, then an analogue of the extension theorem holds for the symmetrized weight composition built on arbitrary group. The result was improved in \cite{assem} where the author showed, in some additional assumptions, the maximality of the cyclic socle condition for the symmetrized weight composition built on the full automorphism group of an alphabet. There the author showed the relation between extension properties with respect to the Hamming weight and the symmetrized weight composition. 

There exist various results on the extension property for arbitrary weight functions, and particularly for homogeneous weights and the Lee weight, as for example in \cite{barra}, \cite{greferath-biinvariant} and \cite{langevin}.

In this paper we give the complete proof of the maximality of cyclic socle condition for the extension theorem in the context of codes over a module alphabet and symmetrized weight composition built on arbitrary group, see \Cref{thm:cor1}. The result is used to show that for a noncyclic socle alphabet and arbitrary weight function the extension property does not hold, see \Cref{thm:anyweight-noncyclic-socle}.

\section{Preliminaries}
Let $R$ be a ring with identity and let $A$ be a finite left $R$-module.
Consider a group $\Aut_R(A)$ of all $R$-linear automorphisms of $A$. Let $G$ be a subgroup of $\Aut_R(A)$. Consider the action of $G$ on $A$ and denote by $A/G$ the set of orbits.

Let $\F(X,Y)$ denote the set of all maps from the set $X$ to the set $Y$.

Let $n$ be a positive integer and consider a module $A^n$.
Define a map $\swc_G: A^n \rightarrow \F(A/G,\Q)$, called the \emph{symmetrized weight composition} built on the group $G$. For each $a \in A^n$, $O \in A/G$,
\begin{equation*}
\swc_G(a)(O) =  \card{ \{i \in \range{n} \mid a_i \in O\} }\;.
\end{equation*}
The \emph{Hamming weight} $\wt: A^n \rightarrow \{0, \dots, n\}$ is a function that counts the number of nonzero coordinates.
There is always a zero orbit $\{0\}$ in $A/G$. For each $a \in A^n$, $\swc_G(a)(\{0\}) = n - \wt(a)$.

Consider a linear code $C \subseteq A^n$ and a map $f \in \Hom_R(C,A^n)$. The map $f$ is called an \emph{$\swc_G$-isometry} if $f$ preserves $\swc_G$. We call $f$ a \emph{Hamming isometry} if $f$ preserves the Hamming weight.

A \emph{closure} of a subgroup $G \leq \Aut_R(A)$, denoted $\bar{G}$, is defined as,
\begin{equation*}
\bar{G}= \{ g \in \Aut_R(A) \mid \forall O \in A/G, g(O) = O\}\;.
\end{equation*}

Obviously, $G$ is a subgroup of $\bar{G}$. If $G = \bar{G}$, then the group is called \emph{closed}. Also, $\swc_G = \swc_{\bar{G}}$, since both groups have the same orbits. More on a closure of a group and its properties see in \cite{wood-aut-iso}.

A map $h:A \rightarrow A$ is called \emph{$G$-monomial} if there exist a permutation $\pi \in S_n$ and automorphisms $g_1, g_2, \dots, g_n \in \bar{G}$ such that for any $a \in A^n$
	\begin{equation*}
		h\left( (a_1, a_2, \dots,a_n) \right) = \left( g_1(a_{\pi(1)}), g_2(a_{\pi(2)}) \dots, g_n(a_{\pi(n)})\right)
	\end{equation*}
	
It is not difficult to show that a map $f \in \Hom_R(A^n, A^n)$ is an $\swc_G$-isometry if and only it is a $G$-monomial map.

We say that $A$ has an \emph{extension property} with respect to $\swc_G$ if for any code $C \subseteq A^n$, each $\swc_G$-isometry $f \in \Hom_R(C,A^n)$ extends to a $G$-monomial map. 

\vspace{1em}

\noindent\textbf{Characters and the Fourier transform.}
Let $A$ be a left $R$-module. The module $A$ can be seen as an abelian group, i.e., $A$ is a $\mathbb{Z}$-module. Consider a multiplicative $\mathbb{Z}$-module $\mathbb{C}^*$. Denote by $\hat{A} = \Hom_\mathbb{Z}(A, \mathbb{C}^*)$ the set of characters of $A$. The set $\hat{A}$ has a natural structure of a right $R$-module, see \cite{greferath} (Section 2.2).

Let $W$ be a left $R$-module. The Fourier transform of a map $f: W \rightarrow \mathbb{C}$ is a map $\mathcal{F}(f): \hat{W} \rightarrow \mathbb{C}$, defined as
\begin{equation*}
\mathcal{F}(f)(\chi) = \sum_{w \in W} f(w)\chi(w)\;.
\end{equation*}
Recall the indicator function of a subset $Y$ of a set $X$ is a map $\id_Y: X \rightarrow \{0,1\}$, such that $\id_Y(x) = 1$ if $x \in Y$ and $\id_Y(x) = 0$ otherwise.
For a submodule $V \subseteq M$, $$\mathcal{F}(\id_V) = \card{V} \id_{V^\perp}\;,$$ where the dual module $V^\perp \subseteq \hat{M}$ is defined as 
$$V^\perp = \{ \chi \in \hat{M} \mid \forall v \in V, \chi(v) = 1 \}\;.$$
Note that the Fourier transform is invertible, $V^{\perp\perp} \cong V$, see \cite{greferath} (Section 2.2), and it is true that for any $R$-submodules $V, U \subseteq W$, $(V \cap U)^\perp = V^\perp + U^\perp$.

\section{Extension criterium}\label{sec:ext-criterium}
Let $W$ be a left $R$-module isomorphic to $C$.
Let $\lambda \in \Hom_R(W, {A^n})$ be a map such that $\lambda(W) = C$.
Present the map $\lambda$ in the form $\lambda = (\lambda_1,\dots, \lambda_n)$, where $\lambda_i \in \Hom_R(W,{A})$ is a projection on the $i$th coordinate, for $i \in \range{n}$.
Let $f: C \rightarrow A^n$ be a homomorphism of left $R$-modules.
Define $\mu = f\lambda \in \Hom_R(W,A^n)$.

An $R$-module $A$ is called \emph{$G$-pseudo-injective}, if for any submodule $B \subseteq A$, each injective map $f \in \Hom_R(B,A)$, such that for any $O \in A/G$, $f(O \cap B) \subseteq O$, extends to an element of $\bar{G}$.

\begin{proposition}\label{thm-isometry-criterium}
	The map $f \in \Hom_R (C,A^n)$ is an $\swc_G$-isometry if and only if for any $O \in A/G$, the following equality holds,
	\begin{equation}\label{eq-main-space-equation}
	\sum_{i=1}^n \id_{\lambda_i^{-1}(O)} = \sum_{i=1}^n \id_{\mu_i^{-1}(O)}\;.
	\end{equation}	
	If $f$ extends to a $G$-monomial map, then there exists a permutation $\pi \in S_n$ such that for each $O \in A/G$ the equality holds,
	\begin{equation}\label{eq-condition}
	\lambda_i^{-1}(O) = \mu_{\pi(i)}^{-1}(O)\;.
	\end{equation}
	If $A$ is $G$-pseudo-injective and \cref{eq-condition} holds, then $f$ extends to a $G$-monomial map. 
\end{proposition}
\begin{proof}
	For any $w \in W$, $O \in A/G$,
	\begin{equation*}
	\swc(\lambda(w))(O) = \sum_{i=1}^n \id_{O}(\lambda_i(w)) = \sum_{i=1}^n \id_{\lambda_i^{-1}(O)}(w)\;.
	\end{equation*}
	Therefore, the map $f$ is an $\swc_G$-isometry if and only if \cref{eq-main-space-equation} holds.
	
	If $f$ is extendable to a $G$-monomial map with a permutation $\pi \in S_n$ and automorphisms $g_1, \dots, g_n \in \bar{G}$, then for all $i \in \range{n}$, $\mu_{\pi(i)} = g_i \lambda_i$. Hence, for all $O \in A/G$, $\mu_{\pi(i)}^{-1}(O) = \lambda_i^{-1} (g_i^{-1}(O)) = \lambda_i^{-1} (O)$.
	
	Prove the last part. Fix $i \in \range{n}$. From \cref{eq-condition} calculated in the orbit $\{0\}$, $\Ker \lambda_i = \Ker \mu_{\pi(i)} = N \subseteq W$. Consider the injective maps $\bar{\lambda}_i, \bar{\mu}_{\pi(i)} : W/N \rightarrow A$ such that $\bar{\lambda}_i(\bar{w}) = \lambda_i(w)$ and $\bar{\mu}_{\pi(i)}(\bar{w}) = \mu_{\pi(i)}(w)$ for all $w \in W$, where $\bar{w} = w + N$.
	
	One can verify that for all $O \in A/G$, $\bar{\lambda}_i^{-1}(O) = \bar{\mu}_{\pi(i)}^{-1}(O)$. Then, it is true that for the injective map $h_i = \bar{\mu}_{\pi(i)}\bar{\lambda}_i^{-1} \in \Hom_R(A, A)$, $h(O \cap \bar{\lambda}_i(W/N)) \subseteq O$, for all $O \in A/G$. Since $A$ is $G$-pseudo-injective, there exists a $G$-monomial map $g_i$ such that $g_i = h_i$ on $\bar{\lambda}_i(W/N) = \lambda_i(W) \subseteq A$. It is easy to check that $\lambda_i = g_i \mu_{\pi(i)}$. Hence $f$ extends to a $G$-monomial map.
\end{proof}


\begin{proposition}
	A module $A$ is $G$-pseudo-injective if and only if for any code $C \subset A^1$, each $\swc_G$-isometry $f \in \Hom_R(C,A)$ extends to a $G$-monomial map.
\end{proposition}
\begin{proof}
	Prove the contrapositive. By definition, $A$ is not $G$-pseudo-injective, if there exists a module $C \subseteq A$ and an injective map $f \in \Hom_R(C,A)$, such that for each $O \in A/G$, $f(O \cap C) \subseteq O$, but $f$ does not extend to an automorphism $g \in \bar{G}$. Equivalently, $\swc_G(x) = \swc_G(f(x))$, for all $x \in C$, yet $f$ does not extend to a $G$-monomial map.
\end{proof}

\begin{remark}
A module $A$ is called \emph{pseudo-injective}, if for any module $B \subseteq A$, each injective map $f \in \Hom_R(B,A)$ extends to an endomorphism in $\Hom_R(A,A)$.
In \cite{dinh-lopez-1} and \cite{wood-foundations} the authors used the property of pseudo-injectivity to describe the extension property for the Hamming weight. They showed that an alphabet is not pseudo-injective if and only if there exists a linear code $C \subset A^1$ with an unextendable Hamming isometry.
\end{remark}

\begin{remark}
	Not all the modules are $G$-pseudo-injective. It is even true that not all pseudo-injective modules are $G$-pseudo-injective, for some $G \leq \Aut_R(A)$. In our future paper we will give a description of $G$-pseudo-injectivity of finite vector spaces. Apparently, despite the fact that vector spaces are pseudo-injective, almost all vector spaces, except a few families, are not $G$-pseudo-injective for some $G$.
\end{remark}

\section{Matrix module alphabet}\label{sec:mm}
Let $R = \M_{k\times k} (\mathbb{F}_q)$ be the ring of $k \times k$ matrices over the finite field $\mathbb{F}_q$, where $k$ is a positive integer and $q$ is a prime power. It is proved in \cite[p.~656]{lang} that each left(right) module $R$-module $U$ is 
isomorphic to $\M_{k \times t}(\mathbb{F}_q)$ ($\M_{t \times k}(\mathbb{F}_q)$), for some nonnegative integer $t$. 
Call the integer $t$ a \emph{dimension} of $U$ and denote $\dim U = t$.

Let $m$ be a positive integer, $m > k$. 
Let $M$ be an $m$-dimensional left(right) $R$-module.
Let $\L(M)$ be the set of all $R$-submodules in $M$. Consider a poset $(\L(M), \subseteq)$ and define a map
$$E : \F(\L(M), \Q) \rightarrow \F(M, \Q)\;,\quad \eta \mapsto \sum_{U \in \L(M)} \eta(U)\id_{U}\;.$$
The set $\F(\L(M),\Q)$ has a structure of an $\card{\L(M)}$-dimensional vector space over the field $\Q$. In the same way, $\F(M, \Q)$ is an $\card{M}$-dimensional $\Q$-linear vector space. The map $E$ is a $\Q$-linear homomorphism.
Similar notions of a \emph{multiplicity function} and the ``\emph{$W$ function}" were observed in \cite{wood-aut-iso} and \cite{wood-foundations}.

Let $V$ be a submodule of $M$. Consider a map in $\F(\L(M),\Q)$,
\begin{equation*}
\eta_V(U) = \left\lbrace\begin{array}{ll}
(-1)^{\dim V - \dim U} q^{\binom{\dim V - \dim U}{2}}& \text{, if } U \subseteq V;\\
0& \text{, otherwise}. 
\end{array}
\right.
\end{equation*}
In fact, $\eta_V(U) = \mu(U,V)$, where $\mu$ is the M\"{o}bius function of the poset $(\L(M),\subseteq)$, see \cite{wood-foundations} (Remark 4.1).

\begin{lemma}\label{lemma:eta-v}
	For any $V \subseteq M$, if $\dim V > k$, then $E(\eta_V) = 0$.
\end{lemma}
\begin{proof}
	First, note that for a submodule $U \subseteq M$ and an element  $x \in M$ the inclusion $x \in U$ holds if and only if for the cyclic module $xR$ the inclusion $xR \subseteq U$ holds. Calculate, for any $x \in M$,
\begin{align*}
E(\eta_V)(x) = \sum_{U \in \L(M)} \eta_V(U) \id_U(x) = \sum_{U \subseteq V} \mu(U,V) \id_U(x) = \sum_{xR \subseteq U \subseteq V} \mu(U,V)\;.
\end{align*}
From the duality of the M\"{o}bius function, see \cite{rota} (Proposition 3),
the last sum is equal to $\sum_{V \supseteq U \supseteq xR} \mu^*(V,U)$, where $\mu^*$ is the M\"{o}bius function of the poset $(\L(M), \supseteq)$.
Since $\dim xR \leq \dim R_R = k < \dim V$, $V \supset xR$. From the definition of the M\"{o}bius function the resulting sum equals $0$.
\end{proof}

Let $\ell$ be a positive integer, $m \geq \ell > k$. Fix a submodule $X$ in $M$ of dimension $m - l$. Define two subsets of $\L(M)$,
$$S_{=\ell} = \{ V \in \L(M) \mid \dim V = \ell, V \cap X = \{0\} \}\;,$$ 
$$S_{<\ell} = \{ V \in \L(M) \mid \dim V < \ell, V \cap X = \{0\} \}\;.$$ Consider a map,
\begin{equation*}
E': \F(S_{=\ell}, \Q) \rightarrow \F(S_{<\ell}, \Q), \quad \xi \mapsto \sum_{V \in S_{=\ell}} \xi(V) \eta_V\;.
\end{equation*}
The map $E'$ is a $\Q$-linear homomorphism of $\Q$-linear vector spaces.
\begin{lemma}\label{lemma:number-nonint}
	Let $A$ be an $a$-dimensional $\mathbb{F}_q$-linear vector space and let $B$ be its $b$-dimension subspace. Then,
\begin{equation*}
\card{ \{ C \subseteq A \mid C \cap B = \{0\}, \dim_{\mathbb{F}_q} C = c\} } = q^{bc} \binom{a-b}{c}_q\;.
\end{equation*}
\end{lemma}
\begin{proof}
	See \cite{dist-reg-graphs} (Lemma 9.3.2).	
\end{proof}

\begin{lemma}\label{lemma:qbin-ineq}
	For any positive integer $t$ there exists an integer $x > t$ such that
\begin{equation*}
\sum_{i=0}^{t-1} \binom{x}{i}_q < q^{t(x-t)}\;.
\end{equation*}
\end{lemma}
\begin{proof}
If $t = 1$ the inequality holds for any $x > 1$. Let $t > 2$ and consider only $x \geq 2t$. Then,
\begin{align*}
\sum_{i=0}^{t-1} \binom{x}{i}_q &< t \binom{x}{t-1}_q = t \frac{(q^x - 1) \dots (q^{x - t + 2} - 1)}{(q^{t-1} - 1) \dots (q - 1)}
\\&< c(t) q^x q^{x-1}\dots q^{x-t + 2}
= c(t)q^{\frac{(t-1)(2x - t + 2)}{2}} = q^{(t-1)x + c'(t)} \;,
\end{align*}
for some constants $c(t), c'(t)$ that depend on $t$. The inequality $q^{(t-1)x + c'(t)} \leq q^{t(x-t)}$ holds for any $x \geq t^2 + c'(t)$. Thus we can take $x$ large enough to be greater than $2t$ and to satisfy the inequality.
\end{proof}

\begin{lemma}\label{lemma:nonzero-kernel}
	There exists $m$ such that $\Ker E' \neq \{0\}$.
\end{lemma}
\begin{proof}
For any positive integer $m$ there exists an isomorphism between the poset of subspaces of an $m$-dimensional vector spaces and the poset of submodules of an $m$-dimensional $R$-module, see \cite{yaraneri}. Therefore, we can use \Cref{lemma:number-nonint} for $R$-modules. 

Calculate the cardinalities of sets. From \Cref{lemma:number-nonint},
\begin{equation*}
\card{S_{=\ell}} = q^{l(m-l)} \binom{l}{l}_q = q^{l(m-l)}\;.
\end{equation*}
Since $S_{<\ell} \subseteq \{ U \in \L(M) \mid \dim U < \ell \}$,
\begin{equation*}
\card{S_{<\ell}} \leq \sum_{i=0}^{\ell - 1} \binom{m}{i}_q\;.
\end{equation*}
From \Cref{lemma:qbin-ineq}, there exists $m$ such that $\card{S_{=\ell}} > \card{S_{<\ell}}$. Therefore
$$\dim_{\Q} \Ker E' \geq \dim_{\Q} \F(S_{=\ell}, \Q) - \dim_{\Q} \F(S_{<\ell}, \Q) = \card{S_{=\ell}} - \card{S_{<\ell}} > 0\;.$$
\end{proof}

From now we assume that $\dim M = m$ is such that $\Ker E' \neq \{0\}$. It is possible due to \Cref{lemma:nonzero-kernel}. Let $0 \neq \xi \in \Ker E' \subseteq \F(S_{=\ell}, \Q)$. Define a map,
\begin{equation*}
\eta(V) = \left\lbrace\begin{array}{ll}
\xi(V)& \text{, if } V \in S_{=\ell};\\
0& \text{, otherwise}. 
\end{array}
\right.
\end{equation*}

\begin{lemma}\label{lemma:eta-zero}
	The equality $E(\eta) = 0$ holds.
\end{lemma}
\begin{proof}
From \Cref{lemma:eta-v}, for any $V \in \L(M)$ of dimension $\ell > k$,
\begin{equation*}
0 = E(\eta_V) = \sum_{U \in \L(M)} \eta_V(U) \id_U = \id_V + \sum_{U \in S_{<\ell}} \eta_V(U) \id_U\;.
\end{equation*}
Then,
\begin{align*}
E(\eta) =& \sum_{V \in \L(M)} \eta(V) \id_V = \sum_{V \in S_{=\ell}} \xi(V) \id_V = - \sum_{V \in S_{=\ell}}\xi(V) \sum_{U \in S_{<\ell}} \eta_V(U) \id_U
\\=& - \sum_{U \in S_{<\ell}}  \left(\sum_{V \in S_{=\ell}}\xi(V)  \eta_V\right)(U) \id_U = - \sum_{U \in S_{<\ell}}  E'(\xi)(U) \id_U = 0\;.
\end{align*}
\end{proof}
One can see that we can choose $\xi \in \Ker E'$ to have integer values, by multiplying by the proper scalar $\lambda \in \Q$, so we assume $\eta$ also has integer values.

The module of characters $W = \hat{M}$ is a right(left) $R$-module of dimension $m$. The poset $(\L(M),\subseteq)$ is isomorphic to the poset $(\L(W),\subseteq)$.
For any module $V \in \L(W)$ the dual module $V^\perp$ is in $\L(M)$ and vice versa.
Define a dual map $\eta^{\perp} \in \F(\L(W), \Q)$ as follows, for any $V \in \L(W)$,
\begin{equation*}
\eta^{\perp}(V) = \eta(V^\perp)\;.
\end{equation*}
Note that the map $\eta^\perp$ has only integer values.

\begin{lemma}\label{lemma:etaperp-zero}
	The equality $E(\eta^{\perp}) = 0$ holds.
\end{lemma}
\begin{proof}
The Fourier transform is a $\Q$-linear map. Note that since for all $V \in S_{=\ell}$, $\dim V = \ell$, $\card{V} = q^{k\ell}$ and $\card{V^\perp} = q^{(m - \ell) k}$. Calculate,
\begin{align*}
\mathcal{F} (E(\eta^{\perp})) &= \sum_{V \in \L(W)} \eta^{\perp}(V)\mathcal{F}(\id_V) = \sum_{V^{\perp} \in S_{=\ell}} \eta(V^\perp) \card{V}\id_{V^{\perp}} 
\\&= \sum_{U \in S_{=\ell}} \card{U^{\perp}} \eta(U)\id_{U} = q^{(m - \ell) k}\sum_{U \in S_{=\ell}} \eta(U)\id_{U} 
\\&= q^{(m - \ell) k} \sum_{U \in \L(M)} \eta(U)\id_{U} = q^{(m - \ell) k} E(\eta)\;.
\end{align*}
From \Cref{lemma:eta-zero}, $E(\eta) = 0$, so $\mathcal{F} (E(\eta^{\perp})) = 0$. The Fourier transform is invertible, thus $E(\eta^{\perp}) = \mathcal{F}^{-1}(0) = 0$.
\end{proof}

Define $S^{\perp} = \{ V \in \L(W) \mid V^\perp \in S_{=\ell} \}$. For each $V \in S^{\perp}$, by duality, $V + X^\perp = W$. Recall $\dim V = m - \ell$ and $\dim X^\perp = m - (m - \ell) = \ell$, and therefore $V \cap X^{\perp} = \{0\}$. Alternatively, $S^{\perp}$ can be defined as,
$$
S^{\perp} = \{ V \in \L(W) \mid V \cap X^\perp = \{0\}, \dim V = m - \ell \}\;.
$$
Note that
\begin{equation*}
\sum_{V \in S^{\perp}} \eta^{\perp}(V) = \sum_{V \in \L(W)} \eta^{\perp}(V)\id_V(0) = E(\eta^{\perp})(0) = 0\;.
\end{equation*}

\begin{proposition}\label{thm:swc_g}
	Let $R = \M_{k\times k} (\mathbb{F}_q)$ and let $A$ be an $\ell$-dimensional $R$-module, where $\ell > k$. Then there exist a positive integer $n$ and an $R$-linear code $C \subset A^n$ with an unextendable $\swc_{G}$-isometry $f \in \Hom_R(C, A^n)$, for any $G \leq \Aut_R(A)$.
\end{proposition}
\begin{proof}
Use the notation of this section.
Since $\dim X^{\perp} = \dim A$ there is a module isomorphism $\psi: X^{\perp} \rightarrow A$. Define the length of the code $n$ as the sum of the positive values $\eta^\perp(U)$, $U \in S^{\perp}$. From the calculations above, it is equal to the sum of the negative values $\eta^\perp(U)$, $U \in S^{\perp}$.

For any $i \in \range{n}$, let us define $\lambda_i \in \Hom_R(W,A)$. Choose a submodule $V_i \subseteq W$, such that $\eta^{\perp}(V_i) > 0$, (do it for the submodule $V = V_i$ exactly $\eta^{\perp}(V)$ times). Since $V_i \in S^{\perp}$, $V_i \cap X^\perp = \{0\}$, $\dim V_i = m - \ell$ and $\dim X^\perp = \ell$, and therefore $W = V_i \oplus X^{\perp}$.
Define 
$$\lambda_i: W = V_i \oplus X^{\perp} \rightarrow A, \quad (v, x) \mapsto \psi(x)\;.$$
Then $\Ker \lambda_i = V_i \subseteq W$ and for any $a \in A$, $\lambda_i^{-1}(a) = \psi^{-1}(a) + V_i$. In the same way, for any $i \in \range{n}$, define maps $\mu_i: W \rightarrow A$ for the modules $U_i$, such that $\eta^{\perp}(U_i) < 0$.

Check \cref{eq-main-space-equation} is satisfied for the trivial group $\{e\} < \Aut_R(A)$. The trivial group has one-point orbits in $A$. Calculate, using \Cref{lemma:etaperp-zero}, for any $a \in A$, for any $w \in W$,
\begin{align*}
\sum_{i = 1}^n (\id_{\lambda_i^{-1}(\{a\})} - \id_{\mu_i^{-1}(\{a\})})(w) &= \sum_{U \in S^{\perp}} \eta^{\perp}(U) \id_{\psi^{-1}(a) + U}(w)
\\= \sum_{U \in S^{\perp}} \eta^{\perp}(U) \id_{U}(w - \psi^{-1}(a)) &= E(\eta^\perp)(w - \psi^{-1}(a)) = 0\;.
\end{align*}
Also, it is easy to see that for any $i,j \in \range{n}$, $\Ker \lambda_i \neq \Ker \mu_j$.

Since \cref{eq-main-space-equation} holds for the orbit $\{0\}$, it is true that $\Ker \lambda = \bigcap_{i=1}^n \Ker \lambda_i = \bigcap_{i=1}^n \Ker \mu_i = \Ker \mu = N \subset W$. Let $\bar{\lambda}, \bar{\mu}$ be two canonical injective maps $\bar{\lambda}, \bar{\mu} \in \Hom_R(W/N, A^n)$ such that $\bar{\lambda}(\bar{w}) = \lambda(w)$ and $\bar{\mu}(\bar{w}) = \mu(w)$ for all $w \in W$, where $\bar{w} = w + N$.

Use the notation of \Cref{sec:ext-criterium}. Define a code $C \subset A^n$ as the image $\lambda(W)$. Define a map $f \in \Hom_R(C,A^n)$ as $f = \bar{\mu}\bar{\lambda}^{-1}$. It is true that $f\lambda = \mu$.
From \Cref{thm-isometry-criterium}, $f$ is an $\swc_{\{e\}}$-isometry. Therefore, $f$ is an $\swc_{G}$-isometry. However, $f$ does not extend even to an $\Aut_R(A)$-monomial map, since \cref{eq-condition} does not hold for the orbit $\{0\}$.
\end{proof}
\section{Main result}
Let $R$ be a finite ring with identity. Recall several results that appear in \cite{wood-foundations} in order to generalize \Cref{thm:swc_g} for the case of arbitrary module alphabet. For the finite ring $R$ there is an isomorphism,
\begin{equation*}
R/\rad(R) \cong \M_{r_1 \times r_1}(\mathbb{F}_{q_1}) \oplus \dots \oplus \M_{r_n \times r_n}(\mathbb{F}_{q_n})\;,
\end{equation*}
for nonnegative integers $n, r_1, \dots, r_n$ and prime powers $q_1, \dots, q_n$, see \cite{wood-foundations} and \cite{lam} (Theorem 3.5 and Theorem 13.1).

Consider a finite left $R$-module $A$. Since $\soc(A)$ is a sum of simple $R$-modules, there exist nonnegative integers $s_1, \dots, s_n$ such that,
$$ \soc(A) \cong s_1 T_1 \oplus \dots \oplus s_n T_n\;,$$
where $T_i \cong \M_{r_i \times 1}(\mathbb{F}_{q_i})$ is a simple $\M_{r_i \times r_i}(\mathbb{F}_{q_i})$-module, $i \in \range{n}$.

\begin{proposition}[see \cite{wood-foundations} (Proposition 5.2)]\label{thm:noncyclic-socle-property}
	The socle $\soc(A)$ is cyclic if and only if $s_i \leq r_i$, for all $i \in \range{n}$.
\end{proposition}

In \cite{elgarem-megahed-wood} (Theorem 3) the authors proved the following.
\begin{theorem}\label{thm:elgarem-wood}
	Let $R$ be a finite ring with identity. Let $A$ be a finite $R$-module with a cyclic socle. The alphabet $A$ has an extension property with respect to the symmetrized weight composition $\swc_G$, built on any subgroup $G \leq \Aut_R(A)$.
\end{theorem}

We prove the complementary part.

\begin{theorem}\label{thm:noncyclic-socle-swc}
	Let $R$ be a finite ring with identity. Let $A$ be a finite $R$-module with a noncyclic socle. The alphabet $A$ does not have an extension property with respect to the symmetrized weight composition $\swc_G$, built on any subgroup $G \leq \Aut_R(A)$.
\end{theorem}
\begin{proof}
	Our proof repeats the idea of \cite{wood-foundations} (Theorem 6.4).
	From \Cref{thm:noncyclic-socle-property}, $\soc(A)$ is not cyclic, then there exists an index $i$ with $s_i > r_i$. Of course, $s_i T_i \subseteq \soc(A) \subseteq A$. Recall that $s_i T_i$ is the pullback to $R$ of the $\M_{r_i \times r_i}(\mathbb{F}_q)$-module $B = \M_{r_i \times s_i}(\mathbb{F}_q)$. Denote the ring $\M_{r_i \times r_i}(\mathbb{F}_q) = R'$.
	
	Because $r_i < s_i$, \Cref{thm:swc_g} implies the existence of $R'$-linear code $C \subset B^n$ and an $\swc_{\{e\}}$-isometry $f \in \Hom_{R'} (C, B^n)$ that does not extend to a $\Aut_{R'}(B)$-monomial map.
	
	Recall the notation of \Cref{sec:ext-criterium}. Denote $V = \Ker \lambda_1$. Define a subcode $C' = \lambda(V) \subseteq \lambda(W) = C$. The first column of $C'$ is a zero-column. Assume that the code $f(C')$ has a zero-column. Then there exists $i \in \range{n}$ such that $V \subseteq \Ker \mu_i$. The code $C$ is constructed from the map $\eta^\perp$, defined in \Cref{sec:mm}, so $\dim V = \dim A = \dim \Ker \mu_i$ for all $i \in \range{n}$. Also, $\Ker \lambda_i \neq \Ker \mu_j$ for all $i,j \in \range{n}$. Therefore it is impossible to have $V \subseteq \Ker \mu_i$ for some $i \in \range{n}$ and thus $f(C')$ does not have a zero-column.
	
	The projection mappings $R \rightarrow R/\rad(R) \rightarrow R'$ allows us to consider $C'$ and $f$ as an $R$-module and an $R$-linear homomorphism correspondingly.
	
	We have $C' \subset (s_i T_i)^n \subseteq \soc(A)^n \subset A^n$ as $R$-modules. The map $f$ is thus an $\swc_{\{e\}}$-isometry of an $R$-linear code over $A$. Since $\{e\} \leq G$, obviously, $f$ is an $\swc_{G}$-isometry. The codes $C'$ and $f(C')$ have different number of zero columns and hence $f$ does not extend to an $\Aut_R(A)$-monomial map.
\end{proof}

\begin{corollary}\label{thm:cor1}
	Let $R$ be a finite ring with identity. Let $A$ be a finite $R$-module. Let $G$ be a subgroup of $\Aut_R(A)$. The alphabet $A$ has an extension property with respect to $\swc_G$ if and only if $\soc(A)$ is cyclic.
\end{corollary}
\begin{proof}
	See \Cref{thm:elgarem-wood} and \Cref{thm:noncyclic-socle-swc}.
\end{proof}

Let $\omega: A \rightarrow \mathbb{C}$ be a function. For a positive integer $n$ define a weight function $\omega: A^n \rightarrow \mathbb{C}$, $\omega(a) = \sum_{i=1}^n \omega(a_i)$.
Consider a code $C \subseteq A^n$. We say that a map $f \in \Hom_R(C,A^n)$ is an \emph{$\omega$-preserving function} if for any $a \in A^n$, $\omega(a) = \omega(f(a))$.

Let $U(\omega) = \{ g \in \Aut_R(A) \mid \forall a\in A, \omega(g(a))= \omega(a) \}$ be a symmetry group of the weight.
An alphabet $A$ is said to have an extension property with respect to the weight function $\omega$ if for any linear code $C \subseteq A^n$, any $\omega$-preserving map $f \in \Hom_R(C,A^n)$ extends to an $U(\omega)$-monomial map.

\begin{corollary}\label{thm:anyweight-noncyclic-socle}
	Let $R$ be a finite ring with identity. Let $A$ be a finite $R$-module with a noncyclic socle. Let $\omega: A \rightarrow \mathbb{C}$ be arbitrary weight function.
	The alphabet $A$ does not have an extension property with respect to the weight $\omega$.
\end{corollary}
\begin{proof}
	Let $C \subseteq A^n$ be a code and let $f \in \Hom_R(C, A^n)$ be an unextendable $\swc_{U(\omega)}$-isometry that exist due to \Cref{thm:noncyclic-socle-swc}. The map $f$ is then an $\omega$-preserving map and it does not extend to a $U(\omega)$-monomial map.
\end{proof}

\begin{remark}
	The length of the code with an unextendable isometries that is observed in \Cref{thm:noncyclic-socle-swc} and later used in \Cref{thm:anyweight-noncyclic-socle} can be large. For the code observed in \Cref{thm:swc_g} we can give a lower bound $n \geq \prod_{i = 1}^{k} (1 + q^i)$, which is proved in \cite{d4}. In a future paper we will show that for the spacial case $k = 1$ we can give an explicit construction for the code observed in \Cref{thm:swc_g}. The length of the resulting code is $q + 1$. Moreover the resulting unextendable $\swc_G$-isometry can be an automorphism. 
\end{remark}

\begin{remark}
	Unlike the case of the Hamming weight, for arbitrary weight function it is not true that the cyclic socle and pseudo-injectivity conditions lead to the extension property. For example, it is still an open question if the extension property holds for the Lee weight and a cyclic group alphabet, see \cite{barra}. In a recent work of Langevin, see \cite{langevin}, the extension property for the Lee weight with an alphabet that is a cyclic group of prime order is proved.
\end{remark}

\bibliographystyle{acm}
\bibliography{biblio}
\normalsize
\end{document}